\documentclass[12pt]{article}
\usepackage{amsmath,amsfonts,amsthm,amssymb, graphicx,color,hyperref}

\newcommand{\version}{June 7, 2014} 

\usepackage{amsthm,amsfonts,amsmath, amscd}
\usepackage{amssymb,amsmath, dsfont} \usepackage{graphicx}
\swapnumbers
\newtheorem{thm}{THEOREM}[section] \newtheorem{lm}[thm]{LEMMA}

 \theoremstyle{definition}

 \newcommand{\dd}{{\, \rm d}}
\newcommand{\tr}{{\rm Tr}} 
\renewcommand{\|}{{\Vert}} 
 \numberwithin{equation}{section}
\def\dgtw{d_{{\rm GTW}}}

\newcommand{\R}{{\mathord{\mathbb R}}}

\newcommand{\N}{{\mathord{\mathbb N}}}

\begin{document}

\def\tr{{\rm Tr}}

\title{Exponential approach to, and properties of, a
    non-equilibrium steady state in a dilute gas}

\author{\vspace{5pt} Eric A. Carlen$^1$, Joel L. Lebowitz$^{1,2}$ and
  Cl\'ement Mouhot$^{3}$ \\
  \vspace{5pt}\small{$1.$ Department of Mathematics, $2.$ Department
    of Physics,}\\
  [-6pt]\small{Rutgers University, 110 Frelinghuysen Road, Piscataway NJ 08854-8019 USA}\\
  \vspace{5pt}\small{$3.$  DPMMS, University of Cambridge, Wilberforce Road, Cambridge CB3 0WA, UK}\\
}

\date{\version}

\maketitle

\medskip

\centerline{\it Dedicated to Errico Presutti, friend and mentor.}

\let\thefootnote\relax\footnote{ \copyright \, 2014 by the
authors. This paper may be reproduced, in its entirety, for
non-commercial purposes.}

\begin{abstract} We investigate a kinetic model of a system in contact
  with several thermal reservoirs at different temperatures
  $T_\alpha$.  Our system is a spatially uniform dilute gas whose
  internal dynamics is described by the nonlinear Boltzmann equation
  with Maxwellian collisions.  Similarly, the interaction with
  reservoir $\alpha$ is represented by a Markovian process that has
  the Maxwellian $M_{T_\alpha}$ as its stationary state. We prove
  existence and uniqueness of a non-equilibrium steady state (NESS)
  and show exponential convergence to this NESS in a metric on
  probability measures introduced into the study of Maxwellian
  collisions by Gabetta, Toscani and Wenberg (GTW). This
  shows that the GTW distance between the current velocity
  distribution to the steady-state velocity distribution is a Lyapunov
  functional for the system.  We also derive expressions for the
  entropy production in the system plus the reservoirs which is always
  positive.
\end{abstract}

\section{Introduction} \label{intro}

The existence, uniqueness and nature of a non-equilibrium steady state
(NESS) of a system in contact with several reservoirs at different
temperatures and/or chemical potentials continues to be one of the
central problems in statistical mechanics, as is the approach to such
a state.  There are only a few models in which the isolated system
evolves according to classical Hamiltonian mechanics or according to
quantum mechanics for which we have even partial answers to these
questions \cite{BLR,D,LLP}.  In addition to the rather unphysical
models corresponding to harmonic crystals and ideal gases, existence
and uniqueness was proven for systems interacting with soft potentials
in contact with thermal walls \cite{KGI}. The resulting NESS is
spatially non-uniform, and we have little information about its
structure. This is true even for cases in which the system is
described mesoscopically by a one particle distribution $f(x,v,t)$, as
in kinetic theory, where correlatons between particles are neglible.

More is possible to prove for NESS of kinetic systems that are
spatially uniform. Such a system, with one reservoir, but acted upon
by an electric field, is investigated in \cite{CELMR1,CELMR2} and will
be discussed later in this paper.

Here we extend this investigation to the case in which the system is
coupled to several thermal reservoirs at different
temperatures. Remarkably, we find, for the first time we believe, a
Lyapunov functional for such systems. This is described in Sections 2
and 3. Then in Section 4 we consider entropy production for such
systems, and in Section 5 we consider the presence of an external
electric field. Finally, in Section 6 we consider the possibility of
deriving such kinetic models from more microscopic descriptions.

\subsection{Description of the basic model}

As already indicated, we deal in this note with a system 
described by the one-particle probability density $f(v,t)$.  We are
interested in particular in the case in which the evolution is given
by the non-linear Boltzmann equation with pseudo-Maxellian molecular
collisions, for which the collision kernel is
$$
Q(f,f) := \int_{\R^3}\int_{\mathbb S^2}
b\left(\frac{v-v_*}{|v-v_*|}\cdot \sigma\right) \big[ f(v_*')f(v') -
f(v_*)f(v) \big]\dd v_* \dd \sigma\ . 
$$ 
Here $\dd \sigma$ denotes the
uniform probability measure on the sphere, and
$$
v' := \frac{v+ v_* + |v-v_*|\sigma}{2}\qquad{\rm and}\qquad v_*' :=
\frac{v+ v_* - |v-v_*|\sigma}{2}\ .
$$
We assume \emph{Grad's angular cut-off} with $\int_{\mathbb S^2}
b(u\cdot \sigma)\dd \sigma =1$ for any unit vector $u$, so that
\begin{equation}\label{norm} \frac12 \int_{-1}^1 b(s)\dd s =1\ .
\end{equation}

Then we can separate $Q(f,f)$ into its gain and loss terms 
$$
Q(f,f) = Q^+(f,f) - Q^-(f,f)
$$ 
where
\begin{equation}\label{res0} 
Q^+(f,g)(v) := \int_{\R^3}\int_{\mathbb S^2}
b\left(\frac{v-v_*}{|v-v_*|}\cdot \sigma\right) g(v_*')f(v')\dd v_*
\dd \sigma \ , \quad Q^- (f,g) := f \ .
\end{equation}

Observe that the equilibria cancelling this collision operator are
given by the so-called Maxwellian density 
\begin{equation*}
  M_{u,T} (v) := \frac{1}{(2 T \pi)^{3/2}} \exp \left( -
    \frac{|v-u|^2}{2T} \right)
\end{equation*}
with bulk velocity $u \in \R^3$, and temperature $T >0$. We denote
simply $M_T = M_{0,T}$ when $u=0$. 

The system is spatially homogeneous, and is coupled to several thermal
reservoirs, indexed by $\alpha$, at temperatures $T_\alpha$.  To
describe the interaction with the reservoirs, we include in the
evolution equation a term of the form $Q(f,R)$. The simplest example
would have two reservoirs at temperatures $T_1$ and $T_2$ with the
same coupling, in which case
\begin{equation}\label{res2} 
R := \frac12 M_{T_1} + \frac12 M_{T_2} \ . 
\end{equation} 

However, our methods and results do not depend very much on this
particular form of $R$, and for this reason we leave the distribution
$R$ unspecified in much of our discussion.

It will be convenient to choose the time scale so that the total loss
term coming from both $Q(f,f)$ and $Q(f,R)$ is simply $-f$. We can do
this whatever the relative strength of the two collision mechanisms by
making an appropriate choice of the time scale so that in terms of a
parameter $\gamma\in (0,1)$, the evolution equation can be written as
\begin{equation}\label{res} 
\frac{\partial f}{\partial t} = (1-\gamma)
Q(f,f) + \gamma Q(f,R)\
\end{equation} where $R$ be any given probability density on $\R^3$.
Of course, if $R = M_T$ for some $T$, then $M_T$ is the unique steady
state solution of \eqref{res}.  However, if $R$ is given by
\eqref{res2} for $T_1\neq T_2$, then we have no simple expression for
any steady state.

Without loss of generality, we can scale the energy such that
\begin{equation}\label{normal} \int_{\R^3} v R(v)\dd v= 0 \qquad{\rm
and}\qquad \int_{\R^3} |v|^2 R(v)\dd v= 1\ .
\end{equation}

\begin{lm} Let $f_\infty$ be any steady state probability density of
  \eqref{res}. Then, assuming \eqref{normal}, we have
$$\int_{\R^3} v f_\infty (v)\dd v= 0 \qquad{\rm and} 
\qquad    \int_{\R^3} |v|^2 f_\infty (v)\dd v= 1\ .$$
\end{lm}

\begin{proof} Let $f$ be a solution of \eqref{res}. For any test
function $\varphi(v)$, and any two probability densities $f$ and $g$,
we have, by a standard computation
$$
\int_{\R^3} Q(f,g)\varphi (v) \dd v = 
\int_{\R^3\times \R^3 \times \mathbb S^2} 
b(\sigma\cdot k) f(v)g(v_*) \big[ \varphi(v') - \varphi(v) \big]\dd v
\dd v_* \dd \sigma
$$
where $k := |v-v_*|^{-1} (v-v_*)$.

For $\varphi(v) = v$, $\varphi(v') - \varphi(v) = \frac12(v_*-v +
|v_*-v|\sigma)$.  Decomposing $\sigma = (\sigma\cdot k)k +
\sigma^\perp$, we have that
\begin{multline*}
\int_{\R^3\times \R^3 \times \mathbb S^2} b(\sigma\cdot k)
f(v)f(v_*) \big[\varphi(v') - \varphi(v)\big]\dd v \dd v_* \dd \sigma = \\
\left[1 -\frac12 \int_{-1}^1s b(s) \dd s\right] \int_{\R^3\times
\R^3}\left( \frac{v_*-v}{2} \right) f(v)g(v_*)\dd v \dd v_* \ .
\end{multline*} 
Therefore,
\begin{eqnarray*} 
  \frac{{\rm d}}{{\rm d}t}\left( \int_{\R^3} f(t,v)v\dd
    v \right) &=& \gamma \int_{\R^3} Q( f,R)v\dd v\nonumber \\ &=& -
  \frac\gamma2\left[1 -\frac12 \int_{-1}^1s b(s) \dd s\right]\int_{\R^3}
  f(t,v)v\dd v \ . 
\end{eqnarray*} 
By \eqref{norm}, $\left[1 -\frac12 \int_{-1}^1s
b(s) \dd s\right] > 0$, and so the first moment relaxes to zero
exponentially fast.

Likewise, for $\varphi(v) = |v|^2$,
$$\varphi(v') - \varphi(v) = \frac{|v_*|^2 - |v|^2}{2} - \sigma\cdot (v+v_*)\ .$$
Decomposing $\sigma$ as before,
\begin{multline*}
\int_{\R^3\times \R^3 \times \mathbb S^2} b(\sigma\cdot k)
f(v)f(v_*) \big[\varphi(v') - \varphi(v)\big]\dd v \dd v_* \dd \sigma = \\
\left[1 -\frac12 \int_{-1}^1s b(s) \dd s\right] \int_{\R^3\times
\R^3}\left( \frac{|v_*|^2-|v|^2}{2}\right) f(v)g(v_*)\dd v \dd v_* \end{multline*}
Therefore,
\begin{eqnarray} 
\frac{{\rm d}}{{\rm d}t}\left( \int_{\R^3}
f(t,v)|v|^2\dd v \right) &=& \gamma \int_{\R^3} Q( f,R)v\dd v\nonumber \\ &=&
- \frac\gamma2\left[1 -\frac12 \int_{-1}^1s b(s)\dd s\right]\left(
\int_{\R^3} f(t,v)|v|^2 \dd v -1\right)\ .\nonumber
\end{eqnarray} It follows that $\left( \int_{\R^3} f(t,v)|v|^2 \dd v
-1\right)$ relaxes to zero exponentially fast. In any steady state,
these moments must have the limiting value.
\end{proof}

\subsection{The fixed-point equation}

Because we have fixed the time scale so that the total loss term is
simply $f$, the steady state equation can be written as 
$$
f = (1-\gamma)Q^+(f,f) + \gamma Q^+(f,R).
$$
We now follow a method introduced in \cite{CCG} to solve this
equation.

Define the function $\Phi$ from the space of probability densities on
$\R^3$ into itself by
\begin{equation}\label{Phi} \Phi(f) = (1-\gamma)Q^+(f,f) + \gamma
Q^+(f,R)
\end{equation} so that the steady state equation is simply
\begin{equation}\label{fixed} f = \Phi(f)\ .
\end{equation}

We shall show that $\Phi$ is contractive in the
Gabetta-Toscani-Wennberg metric \cite{GTW}, which is the metric
defined as follows: Let $f$ and $g$ be two probability densities on
$\R^3$ with finite second moments such that the first and second
moments are identical. Let $\widehat{f}$ and $\widehat{g}$ denote
their Fourier transforms.  Then
$$
d_{{\rm GTW}}(f,g) := \sup_{\xi\neq 0}\frac{|\widehat f(\xi) -
  \widehat g(\xi)|}{|\xi|^2}\ .
$$

The following contraction lemma gives us existence and uniqueness of
steady states for \eqref{res}.

\begin{lm} 
  For all probability densities $f$ and $g$ with the same first and
  second moments as $R$,
$$
\dgtw(\Phi(f),\Phi(g)) \leq \left(1- \gamma \left[ \frac12 -\frac14
    \int_{-1}^1 sb(s)\dd s\right]\right) \dgtw(f,g)\ .
$$ 
In particular, if $b$ is even,
$$
\dgtw(\Phi(f),\Phi(g)) \leq \left(1 -\frac{\gamma}{2}\right) \dgtw(f,g)\ .
$$
\end{lm}

\begin{proof} Using the Bobylev formula \cite{Boby75,Bobylev},
\begin{equation}\label{bobylev} \widehat{Q^+}(f,g) = \int_{\mathbb S^2}
f(\xi_+)g(\xi_-)b\left( \sigma \cdot \frac{\xi}{|\xi|} \right)\dd \sigma
\end{equation} where
\begin{equation}\label{bobylev2} \xi_\pm = \frac{\xi \pm
|\xi|\sigma}{2}\ .
\end{equation} Note that $|\xi_+|^2+|\xi_-|^2 = |\xi|^2$.

Then we decompose 
$$
\widehat{Q^+}(f,f) - \widehat{Q^+}(g,g) = \widehat{Q^+}(f-g,f) +
\widehat{Q^+}(g,f-g)$$ 
and we deduce
\begin{eqnarray*}
  \dgtw\left(Q^+(f,f), Q^+(g,g)\right) &=& \sup_{\xi\neq 0}
  \frac{\left|\widehat{Q^+}(f-g,f) + \widehat{Q^+}(g,f-g)\right|}{|\xi|^2}\nonumber\\
  &\leq& \sup_{\xi\neq 0} \int_{\mathbb S^2}
  \frac{|\widehat{f}-\widehat{g}|(\xi_+)|\widehat{g}|(\xi_-)}{|\xi|^2}
  b\left (\sigma\cdot \frac{\xi}{|\xi|}\right)\dd \sigma\nonumber\\ 
  & +& \sup_{\xi\neq 0} \int_{\mathbb S^2}
  \frac{|\widehat{g}|(\xi_+)|\widehat{f}-\widehat{g}|(\xi_-)}{|\xi|^2}b\left(\sigma\cdot 
  \frac{\xi}{|\xi|}\right)\dd \sigma \ . 
 \end{eqnarray*} 
 Next, using the definition of $\dgtw(f,g)$ and the fact that
 $\|\widehat g\|_\infty \le 1$,
 \begin{multline*} \int_{\mathbb S^2}
\frac{|\widehat{f}-\widehat{g}|(\xi_+)|\widehat{g}|(\xi_-)}{|\xi|^2}b\left(\sigma\cdot
\frac{\xi}{|\xi|}\right)\dd \sigma \\ = \int_{\mathbb S^2}
\frac{|\widehat{f}-\widehat{g}|(\xi_+)|\widehat{g}|(\xi_-)}{|\xi_+|^2}
\frac{|\xi_+|^2}{|\xi|^2} b\left(\sigma\cdot \frac{\xi}{|\xi|}\right)\dd
\sigma \\ \leq  \dgtw(f,g) \int_{\mathbb S^2}
\frac{|\xi_+|^2}{|\xi|^2} b\left(\sigma\cdot \frac{\xi}{|\xi|}\right)\dd \sigma\
.
 \end{multline*} 
Likewise,
 $$
 \sup_{\xi\neq 0} \int_{\mathbb S^2}
 \frac{|\widehat{g}|(\xi_+)|\widehat{f}-\widehat{g}|(\xi_-)}{|\xi|^2} b\left(\sigma\cdot
 \frac{\xi}{|\xi|}\right)\dd \sigma \leq \dgtw(f,g) \int_{\mathbb S^2}
 \frac{|\xi_-|^2}{|\xi|^2} b\left(\sigma\cdot
   \frac{\xi}{|\xi|}\right)\dd \sigma\ .
$$
Then since $|\xi_+|^2+|\xi_-|^2 = |\xi|^2$, we have that
  $$
\dgtw\left(Q^+(f,f),Q^+(g,g)\right) \leq \dgtw(f,g)\ .
$$
  
  Next, by essentially the same calculation,
 \begin{eqnarray*} 
\dgtw\left(Q^+(f,R), Q^+(g,R)\right) &\leq& \int_{\mathbb S^2}
\frac{|\widehat{f}-\widehat{g}|(\xi_+)|\widehat{R}|(\xi_-)}{|\xi_+|^2}
\frac{|\xi_+|^2}{|\xi|^2} b\left(\sigma\cdot \frac{\xi}{|\xi|}\right)\dd \sigma
\nonumber\\ &\leq& \dgtw(f,g) \int_{\mathbb S^2} \frac{|\xi_+|^2}{|\xi|^2}
b\left(\sigma\cdot \frac{\xi}{|\xi|}\right)\dd \sigma \ .
 \end{eqnarray*} 
 Since $|\xi_+|^2 = \tfrac12(|\xi|^2 + |\xi|(\xi\cdot \sigma))$,
 $$
 \int_{\mathbb S^2} \frac{|\xi_+|^2}{|\xi|^2} b\left(\sigma\cdot
   \frac{\xi}{|\xi|}\right)\dd \sigma = \frac14 \int_{-1}^1 \left(1 +
   s\right)b(s)\dd s \ .
$$ 
Altogether, by the triangle inequality, we have
$$
\dgtw\left(\Phi(f,f), \Phi(g,g)\right) \leq \left((1-\gamma) + \gamma
  \left[\frac14 \int_{-1}^1 \left(1 + s\right)b(s)\dd s\right]\right)
\dgtw(f,g)\ ,
$$ 
which gives the result.

\end{proof}

\begin{thm} Suppose that $R$ satisfies \eqref{normal}. Then there is a
unique steady state solution $f_\infty$ of \eqref{res}. Moreover, if
we define a sequence by $f_0 = R$ and $f_{n} = \Phi(f_{n-1})$ for all
$n\in \N$, then $f_\infty = \lim_{n\to\infty} f_n$, and
\begin{equation}\label{iter} \dgtw(f_n,f_\infty) \leq
\frac{\lambda^n}{1-\lambda} \dgtw(\Phi(R),R)\
\end{equation} where
$$\lambda = \left((1-\gamma) + \gamma \left[\frac14 \int_{-1}^1 \left(1 + s\right)b(s)\dd s\right]\right)  < 1\ .$$
\end{thm}

\begin{proof} This is a direct consequence of the previous lemma and
the contraction mapping theorem. Recall from the proof that
$$\dgtw(f_{n+1},f_{n}) \leq \lambda^n \dgtw(\Phi(R),R)$$
Then by the triangle inequality we obtain the final estimate.

\end{proof}

We remark that the equation \eqref{iter} allows the effective
computation of $f_\infty$. Once can certainly evaluate it numerically,
especially when $b$ is even so that $\lambda = 1-\gamma/2$.

\section{Exponential convergence}

Since for small $h$, and any two solutions $f$ and $g$ of \eqref{res},
\begin{equation*} 
  f(t+h) - g(t+h) = h[\Phi(f(t)) - \Phi(g(t)] +
  (1-h)[f(t)-g(t)] + o(h)\ ,
\end{equation*} it follows from our contraction estimate in the
previous section that provided $f(0)$ and $g(0)$ have the same first
and second moments as $R$,
\begin{multline*}
\frac{{\rm d}}{{\rm d}t} \dgtw(f(t),g(t)) \\ \le -
\left[1 - \left((1-\gamma) + \gamma \left[\frac14 \int_{-1}^1 \left(1
        + s\right)b(s)\dd s\right]\right)\right] \dgtw(f(t),g(t))\ .
\end{multline*}

 In particular, taking $g(t) = f_\infty$, we see that
$\dgtw(f(t),f_\infty)$ decreases to zero exponentially fast.
 
We can dispense with the requirement that the initial data $f(0)$ has
the same first and second moments as $R$ by using the correction
technique introduced in \cite{CGT}. Let us describe briefly this
argument. Let us denote 
\begin{align*}
  \lambda_0 &:= \frac12 \left[ 1 - \frac12 \int_{-1} ^{+1} s b(s) \dd s
  \right] >0, \\ 
  \lambda_1 &:= \left[1 - \left((1-\gamma) + \gamma \left[\frac14 \int_{-1}^1 \left(1
          + s\right)b(s)\dd s\right]\right)\right] >0 \ .
\end{align*}
We define (in Fourier variables)
\begin{equation*}
  \widehat{\mathcal{M}}[f] := \chi(\xi) \sum_{|\alpha| \le 2} \left( \int_{\R^3}
    v^\alpha f(v) \dd v \right) \frac{\xi^\alpha}{\alpha !}
\end{equation*}
where we use the standard notation for multi-indeces $\alpha =
(\alpha_1,\alpha_2,\alpha_3) \in \mathbb N^3$, $|\alpha| = \alpha_1 +
\alpha_2 + \alpha_3$, $v^\alpha = v_1^{\alpha_1} v_2^{\alpha_2}
v_3^{\alpha_3}$, $\alpha ! = \alpha_1 ! \alpha_2 ! \alpha_3 !$, and
where $\chi$ is a compactly supported smooth function that is equal to
one around $\xi=0$. Then if we consider two solutions $f$ and $g$ with
possibly different momentum and energy, we write $D = f- g -
\mathcal M[f-g]$, $S = f+g$,
and obtain
\begin{equation*}
  \partial_t D = (1-\gamma) Q (D,S) + (1-\gamma) Q(S,D) + \gamma Q
  (D,R) - W 
\end{equation*}
with 
\begin{multline*}
  W:= \Big[ \partial_t
    \mathcal M[f-g] + (1-\gamma) Q (\mathcal M[f-g],S) \\ + (1-\gamma)
    Q(S,\mathcal M[f-g]) +
    \gamma Q (\mathcal M[f-g],R) \Big] \ ,
\end{multline*}
and one checks by similar moment estimates as above that
\begin{equation*}
  \left| \widehat W (\xi,t)\right| 
  \le C |\xi|^2 \left( \sum_{|\alpha| \le 2} \left| \int_{\R^3}
      v^\alpha \left[f(v,t) - g(v,t) \right]\dd v \right| \right) \le C' |\xi|^2 e^{-\lambda_0
    t}
\end{equation*}
for some constants $C,C' >0$. We then perform the same contraction
estimate as before since $D$ is now the Fourier transform of a
centered zero-energy function, and obtain
\begin{equation*}
  \sup_{\xi \in \R^3} \frac{|\widehat D(\xi,t)|}{|\xi|^2} \le \sup_{\xi \in
    \R^3} \frac{|\widehat D(\xi,0)|}{|\xi|^2} e^{-\lambda_1 t} +
  C'' e^{-\min(\lambda_0,\lambda_1) t}
\end{equation*}
for some constant $C_1$. Finally we deduce by taking $g=f_\infty$ that
$f$ is converging to the equilibrium $f_\infty$ with exponential
rate, measured in the distance 
\begin{equation*}
\dgtw'(f,g) = \sup_{\xi \in \R^3} \frac{|\widehat f(\xi)-\widehat g(\xi)- \hat{\mathcal
  M}[f-g](\xi)|}{|\xi|^2} + \left|\mathcal M[f-g]\right| \ , 
\end{equation*}
which writes $\dgtw'(f,f_\infty) \le C''' e^{-
  \min(\lambda_0,\lambda_1)t}$ for some constant $C'''>0$. 

\section{Diffusive thermal reservoirs}

In some physical situations it is more appropriate to model the
interaction of a system with reservoirs by an Ornstein-Uhlenbeck
continuous time process rather than the discrete time collision model
that we have considered above.  This leads to a kinetic equation of
the form
\begin{equation}\label{OUR} \frac{\partial}{\partial t}f(v,t) = Q(f,f)
+ \sum_\alpha \eta_\alpha T_\alpha \frac{\partial}{\partial v}
\left[M_{T_\alpha} \frac{\partial}{\partial
v}\left(\frac{f}{M_{T_\alpha}}\right)\right]\
\end{equation} The constant $\eta_\alpha$ sets the strength of the
interaction with the $\alpha$th reservoir.  



Note that in this setting, the evolution equation for several
reservoirs reduces to the evolution equation for a single reservoir,
since
$$
\sum_\alpha \eta_\alpha T_\alpha \frac{\partial}{\partial v}
\left[M_{T_\alpha} \frac{\partial}{\partial
v}\left(\frac{f}{M_{T_\alpha}}\right)\right] = \eta T
\frac{\partial}{\partial v} \left[M_{T} \frac{\partial}{\partial
v}\left(\frac{f}{M_{T}}\right)\right]\ $$ where
  $$\eta = \sum_\alpha \eta_\alpha \qquad{\rm and}\qquad T = \frac{1}{\eta} \sum_\alpha \eta_\alpha T_\alpha\ .$$
  
The effective evolution equation
\begin{equation} \label{effective}
\frac{\partial f}{\partial t} = Q(f,f) + \eta T
\frac{\partial}{\partial v} \left[M_{T} \frac{\partial}{\partial
v}\left(\frac{f}{M_{T}}\right)\right]
\end{equation}
is then easy to analyze since in this case the unique stationary state
is $M_T$, and the relative entropy with respect to $M_T$, i.e.,
$\int_{\R^3} f [ \log f - \log M_T]\dd v$ decreases to zero
exponentially fast.  Likewise,
$$\left[ 3T -  \int_{\R^3}  v^2 f(v,t)\dd v\right] = e^{-2\eta t} \left[ 3T -  \int_{\R^3}  v^2 f(v,0)\dd v\right]\ .$$

In fact, more is true: This evolution also has the contractive
property proved in the previous section for \eqref{res}, now using a
different metric, but one that is equivalent to the GTW metric
\cite{GTW}, namely the $2$-Wasserstein metric.

A theorem of Tanaka \cite{T1,T2} says that the evolution described by
the spatially homogeneous Boltzmann equation for Maxwellian molecules
is contractive in this metric. (We do not describe this metric here,
other than to say that like the GTW metric, it metrizes the topology
of weak convergence of probability measures together with convergence
of second moments, and we refer to the book Villani \cite{V} for
the definition and the proof of this fact.)




As Otto has shown \cite{O}, the evolution described by
\begin{equation*}
  \frac{\partial f}{\partial t} = \eta T
  \frac{\partial}{\partial v} \left[M_{T} \frac{\partial}{\partial
      v}\left(\frac{f}{M_{T}}\right)\right]
\end{equation*}
 is exponentially contractive in the $2$-Wasserstein
metric: If $f$ and $g$ are any two solutions of this equation
$$d_{W_2}(f(\cdot,t),g(\cdot,t)) \leq e^{-\eta (t-s)} d_{W_2}(f(\cdot,s),g(\cdot,s))\ .$$
Together with Tanaka's Theorem for the equation with the collision
operator $Q$ only and a splitting argument (i.e., a non-linear Trotter
product argument), one easily establishes that this same estimate is
valid for solutions of \eqref{effective} and therefore \eqref{OUR}. In
particular, it follows that for all solutions $f(v,t)$ of \eqref{OUR},
\begin{multline} d_{W_2}(f(\cdot,t),M_T) \leq e^{-\eta t}
d_{W_2}(f(\cdot,0),M_T) \leq\\ e^{-(n/T)t} \left( \tfrac12 \int
v^2f(v,0)\dd v + T\right)\ .\end{multline} where we have used a simple
estimate (see \cite{V}) for $d_{W_2}(f(\cdot,0),M_T)$ in the last
inequality.

In conclusion, for the diffusive reservoirs, we have not only a ``free
energy type'' Lyapunov functional, namely the relative entropy with
respect to $M_T$, but also a a different one that is similar in nature
to the one we found in the previous section for \eqref{res}.

\section{Entropy production for thermal reservoirs}

We now consider the entropy production when the reservoirs are
thermal. More precisely, we assume that the time evolution of $f(v,t)$
is given by
\begin{equation}\label{thermal} 
  \frac{\partial f}{\partial t} =
  Q(f,f) + \sum_\alpha \int_{\R^3} [K_\alpha(v,v')f(v') -
  K_\alpha(v',v)f(v)]\dd v' \ .
\end{equation} 
The Markovian rates $K_\alpha(v,v')$ describe collisions with the
thermal reservoir at temperature $T_\alpha = \beta_\alpha^{-1}$
resulting in a transition from $v'$ to $v$, and $Q(f,f)$ is a general
Boltzmann type collision operator, not necessarily of the Maxwellian
type. We assume detailed balance for each reservoir; i.e.,
\begin{equation}\label{detailed} 
\forall \, \alpha, \quad K(v,v')M_{T_\alpha}(v') =
K(v',v)M_{T_\alpha}(v) \ . 
\end{equation}

Equations \eqref{thermal} and \eqref{detailed} thus include
\eqref{res} and \eqref{res2} as special cases. However, the
exponential approach proved for the latter may not hold in this more
general case. In fact, even the existence and uniqueness of a
stationary state for \eqref{thermal} and \eqref{detailed} is not
guaranteed; see e.g. \cite{CELMR1,CELMR2}. On the other hand, the
analysis below applies to the broader class of models for which there
is existence and uniqueness of the NESS, and also carries over to the
case in which $f$ and the $K_\alpha$ depend on position $x\in \R^3$,
though we shall not pursue this here.

The rate of change of the system's Boltzmann gas entropy is given by
\begin{equation}\label{enrate} 
\dot S = -\frac{{\rm d}}{{\rm d}t} \int_{\R^3}
f \log f \dd v = \sigma_B + \sum_\alpha \sigma_\alpha - \sigma_R
\end{equation} 
where $\sigma_B$ is the usual rate of change of the
entropy due to the Boltzmann collision term, which is non-negative and
equal to zero if and only if $f$ is a Maxwellian.  The $\sigma_\alpha$
are given by
\begin{equation}\label{siga} 
  \sigma_\alpha := \frac12\int_{\R^3} \int_{\R^3}
  K_\alpha(v,v')M_\alpha(v')[\nu_\alpha(v,t) - \nu_\alpha(v',t)] \log
  \frac{\nu_\alpha(v,t)}{\nu_\alpha(v',t)} \dd v\dd v' \geq 0
\end{equation} where
$$\nu_\alpha(v,t) = \frac{f(v,t)}{M_\alpha(v)}\ .$$
Finally, $\sigma_R$ is the rate of production of entropy in the
reservoirs
\begin{equation}\label{resrate} 
  \sigma_R := \sum_\alpha \beta_\alpha J_\alpha
\end{equation} with
\begin{equation}\label{resrate2} 
  J_\alpha(t) := \frac12 \int \int
  K_\alpha(v,v')f(v',t)\big[v'^2 - v^2\big]\dd v \dd v'\
\end{equation} 
being the flux of energy into the $\alpha$th reservoir at time $t$.

The total rate of entropy production in the system plus reservoirs is
given by
\begin{equation}\label{totalrate} \sigma = \dot S + \sum_\alpha
\beta_\alpha J_\alpha \geq 0\ .
\end{equation} In the stationary state, $\dot S = 0$, and
\begin{equation}\label{stat}\overline{\sigma} = \sum_\alpha
\beta_\alpha \overline{J_\alpha} = \overline{\sigma}_B +
\sum_\alpha\overline{\sigma}_\alpha\geq 0
\end{equation} where the bars denote quantities computed in the
stationary state $f_\infty$.

There is equality in \eqref{stat} if and only if $\overline{\sigma}_B
= 0$ and $\overline{\sigma}_\alpha = 0$ for each $\alpha$, and this is
the case if and only if the stationary state is a Maxwellian, and all
of the reservoirs have the same temperature.

In the case of equal temperature $\beta_\alpha = \beta$ for all
$\alpha$, $\sigma$ in \eqref{totalrate} is given by
$$
\sigma = \frac{{\rm d}}{{\rm d}t}[ S - \beta \mathcal{E}] = - \beta
\frac{{\rm d}}{{\rm d}t} \mathcal{F}$$ where $ \mathcal{E} = \tfrac12
\langle v^2\rangle$ is the average energy of the system and
$\mathcal{F}$ is a kind of free energy.  $ \mathcal{F}$ is thus a
Lyapunov functional achieving its minimum when $f = M_T$. In fact, it
is just the relative entropy of $f$ with respect to $M_T$.

When the $\beta_\alpha$ are unequal, $\sigma$ is not a time
derivative, and the only Lyapunov functionals we know are $d_{\rm
  GTW}(f,f_\infty)$ (or $d_{W_2}(f,f_\infty)$ in the case of diffusive
reservoirs), and only in the case in which the time evolution is given
by \eqref{res}.

In the stationary state we must of course have $\sum_\alpha
\overline{J}_\alpha =0$. Hence if we have only two reservoirs, then
\begin{equation}\label{hotcold} \overline{\sigma} =
(\beta_1-\beta_2)\overline{J}_1 \geq 0\ ;
\end{equation} i.e., heat flows from the hot to the cold reservoir.

We note that when the system is coupled to reservoirs (with equal or
unequal temperatures), then $S$ need not be monotone non-decreasing as
is evident from the fact that we can start from an initial state with
an entropy that is higher than that of the stationary sate; e.g., a
Maxwellian with a sufficiently high temperature. It is only when
$\sigma_R \leq 0$ that $\dot S$ must be nonnegative.

The above considerations remain valid when the Boltzmann collision
kernel $Q(f,f)$ is replaced by the modified Enskog collision kernel
which is generally considered to be a good approximation for a
moderately dense gas; see \cite{GL} and references provided there.

As noted above, in some physical situations it is more appropriate to
model the interaction of a system with reservoirs by an
Ornstein-Uhlenbeck continuous time diffusion process rather than a
discrete time jump process as in \eqref{thermal}.  Our analysis of
entropy production allows for a more general class of diffusive
reservoirs than did our discussion of the contraction property. In
particular, we can allow velocity dependent diffusion coefficients.
Similarly, the collision term may be modified when the system is not a
dilute gas. We shall therefore write a general kinetic equation in the
form
\begin{equation}\label{ias} 
\frac{\partial f}{\partial t} = Q(f) +
\sum_\alpha \frac{\partial}{\partial v} \left[\eta_\alpha(v)T_\alpha
M_{T_\alpha} \frac{\partial}{\partial
v}\left(\frac{f}{M_{T_\alpha}}\right)\right]\
\end{equation} 
requiring only that $Q(f)$ conserve the energy, momentum and mass of
$f$, and that
$$
-\int_{\R^3} Q(f) \log f \dd v \geq 0
$$
with equality if and only if $f = M_T$ for some $T$. The above
analysis can now be repeated with $\sigma_\alpha$ replaced by
$$
\sigma_\alpha ' := \int_{\R^3} f \eta_\alpha(v)T_\alpha \left|
  \frac{\partial}{\partial
    v}\left(\frac{f}{M_{T_\alpha}}\right)\right|^2 \dd v\ .
$$

\if false In this case the stationary state is given by a Maxwellian
with a reciprocal temperature $\overline{\beta}$ that is a weighted
average of the reciprocal temperatures $\beta_\alpha$ of the
reservoirs.  \fi

 \section{Systems driven by an electric field}

 In the systems considered so far in this paper, a non-equilibrium
 steady state has been maintained by at least two reservoirs with
 energy flowing out of some and into others.  A different sort of
 model is investigated in \cite{CELMR1,CELMR2} which concerns a weakly
 ionized plasma with energy supplied by an electric field $E$, and
 removed into a reservoir by a damping mechanism. In this model, the
 probability density $f(v,t)$ evolves according to
\begin{equation}\label{plasma} 
\frac{\partial f}{\partial t} =
\frac{1}{\epsilon}Q(f) -E\cdot \frac{\partial }{\partial v} f + \nu \left[
\tilde f - f\right] + \frac{\partial }{\partial v}\left[ D(v) M
\frac{\partial }{\partial v} \left(\frac{f}{M}\right)\right]\ .
\end{equation}
In \eqref{plasma}, $E$ is a constant electric field, $\epsilon$ is a
small parameter setting the rate of internal collisions, $\tilde{f}$
is the spherical average of $f$ (hence radial, but with the same
energy distribution as $f$), $\nu$ is a constant, and $M$ is a
centered Maxwellian with $T=1$.  See (2.1-4) in \cite{CELMR1}.

When $E= 0$, the unique equilibrium is $M$, and the free energy
$\mathcal{F}$ is a Lyapunov function.  Also when $D(v) = D$
independent of $v$, and $\nu = 0$, the NESS is given by the shifted
Maxwellian $M(v-u)$ where $u = ET/D$. In this case
$$\mathcal{F} = \frac12 \langle |v-u|^2\rangle -\beta S$$
serves as a Lyapunov function governing convergence to the NESS.

However, when $D(v)$ has certain properties, $E\neq 0$ and $\nu >0$,
it is shown in \cite{CELMR1,CELMR2} that for all $\epsilon$
sufficiently small, there is a range of the parameters for which there
are multiple stable stationary solutions of \eqref{plasma}. This means
in particular that there does not exist any global Lyapunov function
for \eqref{plasma} for general parameters.

Finally we have in this stationary state that
$$
\overline{\sigma} = \beta \overline{j}\cdot E \quad \mbox{ where}
\quad \overline{j} = \int_{\R^3} v f_\infty(v)\dd v 
$$ 
is the steady state current.
 
\section{Microscopic models}

We have investigated here the approach to the NESS and some properties
of that state for models in contact with several thermal reservoirs at
different temperatures. We showed that there exist in some cases
Lyapunov functionals even when we do not know the NESS in explicit
form. It would be nice to show that these models correspond to
$N$-particle microscopic models in some suitable scaling limits.

A plausible conjecture is the following: Suppose that the dynamics of
the isolated system is given by a Hamiltonian microscopic dynamics,
and it yields, in some suitable limit, a Boltzmann equation for the
single particle distribution; e.g., hard-sphere collisions under the
Boltzmann-Grad scaling limit \cite{L,PSS}.  Adding stochastic
interactions with thermal reservoirs should then lead, in a suitable
limit, to \eqref{thermal}.

However, this is beyond our current reach, even for the case in which
the system is in contact with a single reservoir.  If one drops the
requirement that the microscopic dynamics be Hamiltonian, the
situation is much better. Such a microscopic derivation was proven
recently by Bonetto, Loss and Vaidyanathan \cite{BLV} when the
isolated system dynamics is given by the Kac stochastic collision
model and there is a single thermal reservoir. One may expect a
similar result to be valid for the Kac system in contact with several
reservoirs.

\medskip

\noindent{\bf Acknowledgements} Work of E.A.C. is partially supported
by N.S.F. grant DMS 1201354.  J.L.L. wishes to thank the I.A.S. its
hospitality during the course of the work, and his work is partially
supported by N.S.F. grant DMR 1104500 and AFOSR grant FA9550.
C.M. wishes to thank the IAS for its support during his visit in may
2014 when this work was done. His research is also partially funded by
the ERC Starting Grant MATKIT.

\end{document}